\documentclass[prx,notitlepage,twocolumn,superscriptaddress,floatfix,nofootinbib]{revtex4-2}

\usepackage{cmap} 

\usepackage[utf8]{inputenc}
\usepackage{amsmath}
\usepackage{amssymb}
\usepackage{MnSymbol}
\usepackage[dvipsnames]{xcolor}
\usepackage{graphicx}
\usepackage{amsbsy} 
\usepackage{soul}

\usepackage{amsthm}
\usepackage[colorlinks=true,citecolor=blue,linkcolor=black]{hyperref}
\newtheorem{theorem}{Theorem}[section]
\newtheorem{corollary}{Corollary}[theorem]
\newtheorem{lemma}[theorem]{Lemma}
\usepackage{scalerel}

\newcommand{\N}[1]{#1^{\scaleto{(N)}{5pt}}}

\newcommand{\new}[1]{\textcolor{black}{#1}}
\newcommand{\newv}[1]{\textcolor{black}{#1}}

\newcommand{\bs}[1]{\boldsymbol{#1}}

\DeclareMathOperator*{\E}{{\mathbb{E}}}

\begin{document}

\title{\new{Precision Bounds on Continuous-Variable State Tomography using Classical Shadows}}
\author{Srilekha Gandhari} 
\affiliation{Joint Center for Quantum Information and Computer Science, NIST/University of Maryland, College Park, Maryland 20742, USA}
\author{Victor V. Albert}
\affiliation{Joint Center for Quantum Information and Computer Science, NIST/University of Maryland, College Park, Maryland 20742, USA}
\affiliation{National Institute of Standards and Technology, 100 Bureau Dr, Gaithersburg, MD 20899, USA}
\author{Thomas Gerrits}
\affiliation{National Institute of Standards and Technology, 100 Bureau Dr, Gaithersburg, MD 20899, USA}
\author{Jacob M. Taylor}
\affiliation{Joint Center for Quantum Information and Computer Science, NIST/University of Maryland, College Park, Maryland 20742, USA}
\affiliation{National Institute of Standards and Technology, 100 Bureau Dr, Gaithersburg, MD 20899, USA}
\affiliation{Joint Quantum Institute, NIST/University of Maryland, College Park, Maryland 20742, USA}
\author{Michael J. Gullans}
\affiliation{Joint Center for Quantum Information and Computer Science, NIST/University of Maryland, College Park, Maryland 20742, USA}
\affiliation{National Institute of Standards and Technology, 100 Bureau Dr, Gaithersburg, MD 20899, USA}

\date{\today}

 \begin{abstract}
    Shadow tomography is a framework for constructing succinct descriptions of quantum states \new{using randomized measurement bases}, called classical shadows, with powerful methods to bound the estimators used. 
    We recast existing experimental protocols for continuous-variable \new{ quantum state} tomography in the classical-shadow framework, obtaining rigorous bounds on the \new{number of independent measurements needed} for estimating density matrices from these protocols.
    We analyze the efficiency of homodyne, heterodyne, photon number resolving (PNR), and photon-parity protocols.
    To reach a desired precision on the classical shadow of an $N$-photon density matrix with a high probability, we show that homodyne detection requires an order $\mathcal{O}(N^{4+1/3})$ measurements in the worst case, whereas PNR and photon-parity detection require $\mathcal{O}(N^4)$ measurements in the worst case (both up to logarithmic corrections).
    We benchmark these results against numerical simulation as well as experimental data from optical homodyne experiments. We find that numerical and experimental homodyne tomography significantly outperforms our bounds, exhibiting a more typical scaling of the number of measurements that is close to linear in $N$.
    We extend our single-mode results to an efficient construction of multimode shadows based on local measurements.
 \end{abstract}
 \maketitle


\section{Introduction}

The ability to estimate states accurately with as few measurements as possible --- the primary goal of quantum state tomography --- yields an indispensable tool in quantum information processing~\cite{NielsenChuang}. 
State characterization is necessary for realizing quantum technologies in finite-dimensional tensor-product quantum systems governed by discrete variables (DV), as well as in quantum systems governed by continuous variables (CV)~\cite{Review_Braunstein}, such as electromagnetic or mechanical modes.

The recent development of classical-shadow tomography yields a succinct way to learn information about a DV quantum state through \new{randomly chosen measurements}, in such a way that the learned information can later be used to predict  properties of the state~\cite{aaronson2018shadow,huang_kueng_preskill_2020,Elben2022}. 
A key benefit of shadow tomography is that it comes with rigorously proven guarantees on the minimum number of samples required to achieve high accuracy with high probability. 
However, the best guarantees require structures from finite-dimensional DV spaces, such as state and unitary designs~\cite{huang_kueng_preskill_2020} or SIC-POVMs~\cite{Stricker2022}, thereby obscuring any practical extension to the intrinsically infinite-dimensional CV systems.

In this paper, we apply the shadow tomography framework to a large family of well-known and well-utilized CV tomographic protocols.
We distill the mathematical tools behind shadow guarantees in such a way that eliminates dependence on strictly DV ingredients and allows us to reformulate established CV protocols in the shadow framework.
This reformulation yields accuracy guarantees for expectation values of local observables whose required number of samples scales polynomially with both the number of participating CV modes and the maximum occupation (\textit{a.k.a.}~photon) number of each mode.

CV tomography is a long-standing and well-developed field~\cite{Raymer_Lvovsky}. Focusing on established protocols allows us to boost their credibility, as opposed to developing new protocols that may be equally  efficient theoretically but whose practical utility is left as an open question (e.g.~\cite[Sec.~VI.A]{Iosue22}).
Specifically, our work focuses on tomographic methods that can be easily implemented experimentally with existing quantum optical technology, which underpins fiber-based and free-space quantum communication and key distribution~\cite{Review_Braunstein}. 
To illustrate the connection to optics and our commitment to the mantra of classical shadows ``measure first, ask questions later''~\cite{Elben2022}, we apply our framework retroactively to experimental quantum-optical data published in 2010~\cite{NIST_experiment}.  
Naturally, our theoretical guarantees apply equally well if the corresponding protocols are performed in other CV platforms such as microwave cavities coupled to superconducting qubits~\cite{girvin2014circuit}, motional degrees of freedom of trapped ions~\cite{DeNeeve2022}, as well as optomechanical~\cite{Aspelmeyer2014} and nano-acoustic~\cite{MacCabe2020} resonators.

\begin{table*}[t]
\centering
\begin{tabular}{|c c c c|} 
 \hline
 Protocol & Basis Expansion & Upper Bound Scaling & Numerics \\ [0.5ex] 
 \hline\hline
 Single-mode Homodyne shadows & Position states/Pattern functions & $N^{4+1/3}$ & $N^2$  \\ [1ex]
 Single-mode PNR shadows & Displaced Fock states/$T$-operator & $N^4$ & $N^4$ \\ [1ex]
 Design based CV tomography \cite{Iosue22} & 3-design & $N^4$ & - \\ [1ex]
 Multimode shadows & Local shadows & $N^{4 M} - N^{(6+1/3)M}$  & - \\ [1ex] 
 \hline
\end{tabular}
\caption{\newv{A summary of the different protocols that exist, to our knowledge, studying the sample complexities of CV state tomography. For each protocol, we indicate the measurement basis and a bound on the scaling of sample complexity (with system size $N$ and number of modes $M$ for multimode shadows). The sample complexity of design based shadows was estimated using known variance bounds  \cite{huang_kueng_preskill_2020, Iosue22}.   Additionally, for single-mode measurement methods considered in this paper, we provide the scaling observed in numerical simulations.}}
\label{tab:protocols}
\end{table*}

The first method that we are able to recast in the shadow-tomography framework is homodyne detection~\cite{Vogel89,Smithey93,LEONHARDT199589,bisio,Raymer_Lvovsky,D_Ariano_1995,Rosati2022}. We show that the number of samples needs to scale at most as the fifth power of the maximum occupation number (up to a logarithmic correction) to yield reliable \textit{homodyne shadow} estimates of a single-mode state. Since we require a finite occupation-number cutoff, this bound only holds for portions of states supported on finite-dimensional subspaces of the infinite-dimensional Fock space. Our theoretical guarantees use several important technical bounds proved in an earlier work~\cite{Guta_paper}. While there have been studies of the statistical efficiency of homodyne tomography (e.g.,~\cite{PhysRevA.57.5013,guta2003,Guta_paper,Albini2009,Alquier2013}), to our knowledge, \new{our bounds are an improvement over previous results as they allow us to analyze the sample complexity for the convergence of the operator norm of the estimated state to its ideal value.  These improvements are made possible because of our use of matrix concentration inequalities reviewed in Sec. II (see Ref.~\cite{Tropp2016,Tropp}) that have not been employed in past work on CV state tomography.}  
A recent paper has also considered the statistical efficiency of homodyne using these improved matrix concentration inequalities \cite{Becker22}.  Here, we present more explicit formulas for the bounds in the case of states with a hard photon number cutoff.  Adapting our approach to the formulation in \cite{Becker22} also leads to explicit bounds in their case that appear to give agreement between the two papers in areas where there is overlap.  
\newv{On a related note, a recent paper Ref.~\cite{Rosati2022} studies the sample complexity of estimating unknown Gaussian or Generalized (i.e. linear combinations of) Gaussian CV states, where a risk function/error is minimized. The distance between the outcomes of the original and estimated state, averaged over channels and measurements, was the metric of error/risk function. Our work, on the other hand, considers an absolute error via the infinity norm of the difference of the original and estimated state. }
\new{Finally, we remark that our results provide upper bounds on the scaling of the number of samples required, i.e. the sample complexity, for estimating the state from random measurements.  We leave open the important problem of obtaining lower bounds on the optimal sample complexity \cite{huang_kueng_preskill_2020}, which have been examined, e.g., in the case of homodyne tomography \cite{Leonhardt96}.}

Building off the formalism in \cite{Wunsche_1991}, we are also able to recast a large class of protocols utilizing displaced Fock states and make contact with heterodyne, photon-number resolved (PNR), and photon-parity tomography \cite{Raymer_Lvovsky}. For the latter two methods, we show that the number of samples needs to scale at most as the fourth power of the maximum occupation (up to a logarithmic correction) to yield reliable \textit{PNR shadow} estimates.  For heterodyne measurements, our upper bound is effectively exponential in the maximum occupation.

Following our theoretical analysis of single-mode systems, we provide a generalization to multimode states.  We prove that local CV shadow data can be used to reconstruct $k$-local reduced density matrices on any subset of the modes with a sample complexity, polynomial in the total number of modes and exponential in $k$.

We support our theoretical analysis with multiple numerical simulations and the analysis of existing experimental homodyne data. We begin by verifying that both the shadow methods indeed construct reliable estimates of a target state.  We then determine how the number of samples required for both the methods scales with the maximum occupation number \(N\). Through simulations involving the reconstruction of the vacuum state, coherent state superposition (CSS) states (i.e., cat states) and random pure states, we observe that the scaling for PNR and photon-parity measurements is similar to the calculated upper bound, whereas homodyne measurements exceed theoretical expectations and surpass the observed PNR scaling in every case.  
We also provide a basic demonstration of homodyne tomography on multimode separable and entangled states, finding that the latter require more samples to reach the same precision.  \newv{Our results are summarized in Table \ref{tab:protocols}.}

We begin with a general result for estimating density matrix elements of a single-mode CV state using a generic shadow-based protocol in Sec.~\ref{section:single_mode_result}. Its application to homodyne measurements and photon-number resolving measurements is discussed in Sec.~\ref{section:applications}. In Sec.~\ref{section:multimode}, we extend our theory to multimode CV states, and support our methods by numerical results and analysis of experimental data in Sec.~\ref{section:numerical_results}. We conclude in Sec.~\ref{section:conclusions}.

\section{Single-mode Shadows:
~~~~~~~~~General case}\label{section:single_mode_result}

Here, we distill the necessary mathematical tools from DV shadow-based accuracy guarantees in a general lemma that we apply to specific CV protocols in later sections. We adapt the slight reformulation from~\cite[Appx. A]{huang2021provablyefficient} of the original shadow work, in which guarantees are shown exclusively in terms of the system state. This allows us to derive sampling guarantees that are suboptimal relative to the qubit case, but are sufficiently general to be applied to numerous CV tomographic protocols.  \new{We assume loss-free propagation and  detection for our theoretical derivations.}

A single-mode CV density matrix can be represented in various ways, e.g., using Wigner, \(Q\), \(P\)~\cite{HILLERY1984121,Ferrie_2011} or more general representations~\cite{Wunsche_1991}, moment expansions~\cite{PhysRevA.54.5291}, homodyne measurements~\cite{LEONHARDT199589,bisio,Raymer_Lvovsky}, or outer products of occupation-number (\textit{a.k.a.} Fock) or other simple basis states. 
All of these ways effectively express the density matrix as a linear combination of some basis set of operators
$\{\sigma(\mu)\}$ with corresponding coefficients \(p(\mu)\), where \(\mu\) parameterizes some generally continuous and non-compact space \(\cal X\),
\begin{equation}\label{eq:rho-expansion}
    \rho=\int_{\cal X} d\mu \, p(\mu)\, \sigma(\mu)~.
\end{equation}

Assuming a non-negative and normalized weight function \(p(\mu)\) and sufficiently well-behaved space \(\cal X\) allows one to re-express the above formula in a statistical fashion, i.e., as an expectation value over samples \(\sigma\), sampled from \(\cal X\) according to the probability distribution governed by \(p(\mu)\),
\begin{equation}\label{eq:rho-stat}
    \rho = \E_{\mu\sim\cal X} \sigma(\mu) = \E_{\mu} \sigma(\mu)~.
\end{equation}
\new{We remark that classical shadow tomography in our formulation refers to tomographic protocols with the restriction that $p(\mu)$ is a probability distribution over random measurement settings and their outcomes.  A general quantum state tomography protocol based solely on Eq.~\eqref{eq:rho-expansion} may not have a formulation as 
a type of classical shadow tomography.  Crucially, though, we will see below that there is a natural analysis of homodyne tomography that fits into the classical shadow framework.}

The above formulation allows for a simple application of matrix concentration inequalities (loosely speaking, the matrix version of the law of large numbers; see Refs.~\cite{Tropp2016,Tropp}). If we are given a set of $T$ independent samples ${\cal{T}}=\{\mu_1,\ldots,\mu_T\}$, 
we can construct an estimator for the state by averaging the basis operators corresponding to the outcomes in the set,
\begin{equation}
    \sigma_T=\frac{1}{T}\sum_{i=1}^T \sigma(\mu_i).
\end{equation}
We define the above estimator as a \textit{shadow of size T}.
Such shadows recover the state by taking their expectation value over all sample sets \(\cal T\),
\begin{equation}\label{eq:sample-average}
    \E_{{\cal T}}~\sigma_T = \frac{1}{T}\sum_{i=1}^T \E_{{\mu_i}} \sigma(\mu_i) = \E_{{\mu}}~\sigma(\mu) = \rho~.
\end{equation}
This convergence property is a necessary starting point for shadow estimation, but it says nothing about \textit{how fast} the convergence happens as the number of samples increases \new{\cite{terminology}.}  

In order to bound the rate of convergence of the above estimate, we restrict \(\rho\) to finite-dimensional truncated versions \(\N{\rho}\), with the dimension cutoff \(N\) assumed to be sufficiently large so that the truncated operators capture the essence of the original ones. 
The resulting matrices have non-zero entries only for row and column indices $0$ through $N-1$. 
A ``regularized'' shadow is constructed with snapshots projected into the $(N-1)$-photon subspace, using $P_N=\sum_{i=0}^{N-1}|i\rangle \langle i|$,
\begin{equation}
    \N{\sigma_T}\equiv P_N \, \sigma_T \,P_N = \frac{1}{T}\sum_{i=1}^T P_N\,\sigma(\mu_i) \,P_N.
\end{equation}
Our analysis can be readily extended to more general soft-cutoffs by replacing the projectors onto fixed photon number spaces by operators with smooth cuttoffs \cite{smoothcutoff}. 
Note also that we do not re-normalize the above regularized density matrices or shadows. Doing so introduces additional factors that would complicate our calculations. Hence, for the sake of ease and uniformity, we chose to work with the regularized matrices as they are.

We now bound the difference between the sample average \(\sigma^{(N)}_T\) and the actual average \(\rho^{(N)}=P_N \rho P_N\), backing out the minimum shadow size \(T\) required for an accurate estimate. 
We employ the infinity norm of the difference $\N{\sigma_T}-\N{\rho}$ to quantify the error. 
\begin{lemma}\label{Lemma_1}(a)
Fix $\epsilon, \delta \in (0,1)$, and let $\N{\sigma_T}$ be an $N$-dimensional classical shadow of a single-mode continuous-variable quantum state $\rho$. If the size of the classical shadow is at least $T$, such that
\begin{equation}\label{eq:general-bound}
T = \frac{2N^2\left(\nu^2+R\epsilon/2N\right)}{\epsilon^2}\left(\log\,2N+\log\,1/\delta\right),
\end{equation}
where 
\begin{subequations}\label{eq:shadow-norms}
\begin{align}
    \nu^{2}&=\left\Vert \E_{{\cal T}}\,\left(\N{\sigma}_{T}\right)^{2}\right\Vert _{\infty}=\left\Vert \E_{\mu}\,\left(\N{\sigma}(\mu)\right)^{2}\right\Vert _{\infty}\label{eq:shadow-norms-nu}\\
    R &\geq\left\Vert \N{\sigma}(\mu)-\E_{\mu}\N{\sigma}(\mu)\right\Vert _{\infty},\label{eq:shadow-norms-R}
\end{align}
\end{subequations}
then the probability
\[ {\rm Pr}\left( \|\N{\sigma_T}-\N{\rho}\|_\infty \leq \epsilon \right) \geq 1-\delta. \]
\end{lemma}
\begin{proof}
Applying the matrix Bernstein inequality~\cite{Tropp}\cite[Appx. A]{huang2021provablyefficient} to $\N{\sigma_T}-\N{\rho}$ and using Eq. (\ref{eq:sample-average}) yields
\[ {\rm Pr}\left(\|\N{\sigma_T}-\N{\rho}\|_\infty\geq\tilde{\epsilon}\right)\leq 2N \,\exp\left(-\frac{T\tilde{\epsilon}^2/2}{\nu^2+R\tilde{\epsilon}/3}\right)~.    \]
Using the equivalence relation $\|X\|_\infty\leq\|X\|_1\leq N \|X\|_\infty$ for an $N$ dimensional system, we have ${\rm Pr}(\|X\|_1\geq\epsilon) \leq {\rm Pr}(N\|X\|_\infty\geq\epsilon)$. Substituting  $\Tilde{\epsilon}/N\rightarrow \epsilon$ we get 
\[  {\rm Pr}\left(\|\N{\sigma_T}-\N{\rho}\|_\infty\geq\epsilon\right)\leq 2N \,\exp\left(-\frac{T\epsilon^2}{2N^2(\nu^2+R\epsilon/3N)}\right).    \]
Equating the error term on the right-hand side of the inequality to $\delta$ and rearranging yield the result (\ref{eq:general-bound}).
\end{proof} 
\newv{We substituted the original variance term $\Vert\mathbb{E}_\mu\left(\N{\sigma}-\N{\rho}\right)^2\Vert_\infty$ in Eq.~\ref{eq:shadow-norms-nu} , which measures the fluctuations about the mean $\N{\rho}$, with the term $\Vert\mathbb{E}_\mu \left(\N{\sigma}\right)^2\Vert_\infty$ for ease in calculations. We noticed that this leads to the derived bounds being higher than they would otherwise in some cases.}

We see that the scaling of the required sample size in Eq. (\ref{eq:general-bound}) depends logarithmically on the inverse error probability \(\delta\) and polynomially on the accuracy \(\epsilon\), recovering the corresponding portion of the qubit shadow result~\cite{huang_kueng_preskill_2020}. To determine dependence on occupation number cutoff \(N\), we need to know the corresponding scaling of \(R\) and \(\nu\) (\ref{eq:shadow-norms}), which are sometimes referred to as \textit{shadow norms}. This scaling depends on the detail of the protocol, and we proceed to apply this lemma to specific cases.

\section{Applications}\label{section:applications}
We apply the general shadow bound from Sec. \ref{section:single_mode_result} to a large family of CV tomographic protocols that include homodyne and photon-number resolving tomography.

\subsection{Homodyne Shadows}
Homodyne tomography is a popular technique used in quantum optics to make quadrature measurements of CV quantum states~\cite{Raymer_Lvovsky}. 
This technique allows one to measure a quadrature operator such as oscillator position, momentum, or any of their superpositions.

Let \(a,a^\dagger\) be the canonical mode lowering and raising operators satisfying \([a,a^\dagger]=1\).
A general quadrature operator,
\begin{equation}
    X_\theta={\textstyle \frac{1}{\sqrt{2}}} ( a e^{-i\theta}+a^\dagger e^{i\theta})~,
\end{equation}
is parameterized by \(\theta\in[0,\pi)\), where \(\theta=0\) (\(\theta=\pi/2\)) corresponds to the pure position (momentum) quadrature.  \new{We consider all possible values of $\theta$ to be admissible and do not consider the problem of choosing an appropriate discretization. In practice, we use pseudorandom number generators to generate the values of $\theta_i$.}
Each operator admits its own set of orthogonal but non-normalizable ``eigenstates'' \(|x_\theta\rangle\) with \(x_\theta\in\mathbb R\). 
Even though such ``eigenstates'' are not quantum states but are instead distributions, they do resolve the identity for each \(\theta\) and thus yield valid positive operator-valued measures (POVMs)~\cite{Holevo2011}.
For a given \(\theta\), a homodyne protocol yields a measurement in the corresponding set of eigenstates.

To perform a homodyne measurement, the unknown system state interferes with an ancillary mode in a coherent state \(|\alpha_{\text{LO}}\rangle\) (coherent local oscillator or LO) via a beam splitter, as shown in Fig.~\ref{fig:homodyne}(a). 
The phase of the coherent state LO is precisely the phase $\theta$ which defines the quadrature that is measured, and is picked uniformly from \([0,\pi)\). 
Both output amplitudes of the beam splitter, $N_1$ and $N_2$ respectively, are measured via homodyne detectors. Their difference $N_-$ is proportional to the quadrature amplitude of the signal:
\begin{equation}
N_-=\sqrt{2}\,|\alpha_{LO}|\,x_\theta \quad \text{with probability} 
    \quad\langle x_\theta|\rho|x_\theta\rangle~.
\end{equation}

\begin{figure}[ht]
    \includegraphics[width=0.47\textwidth]{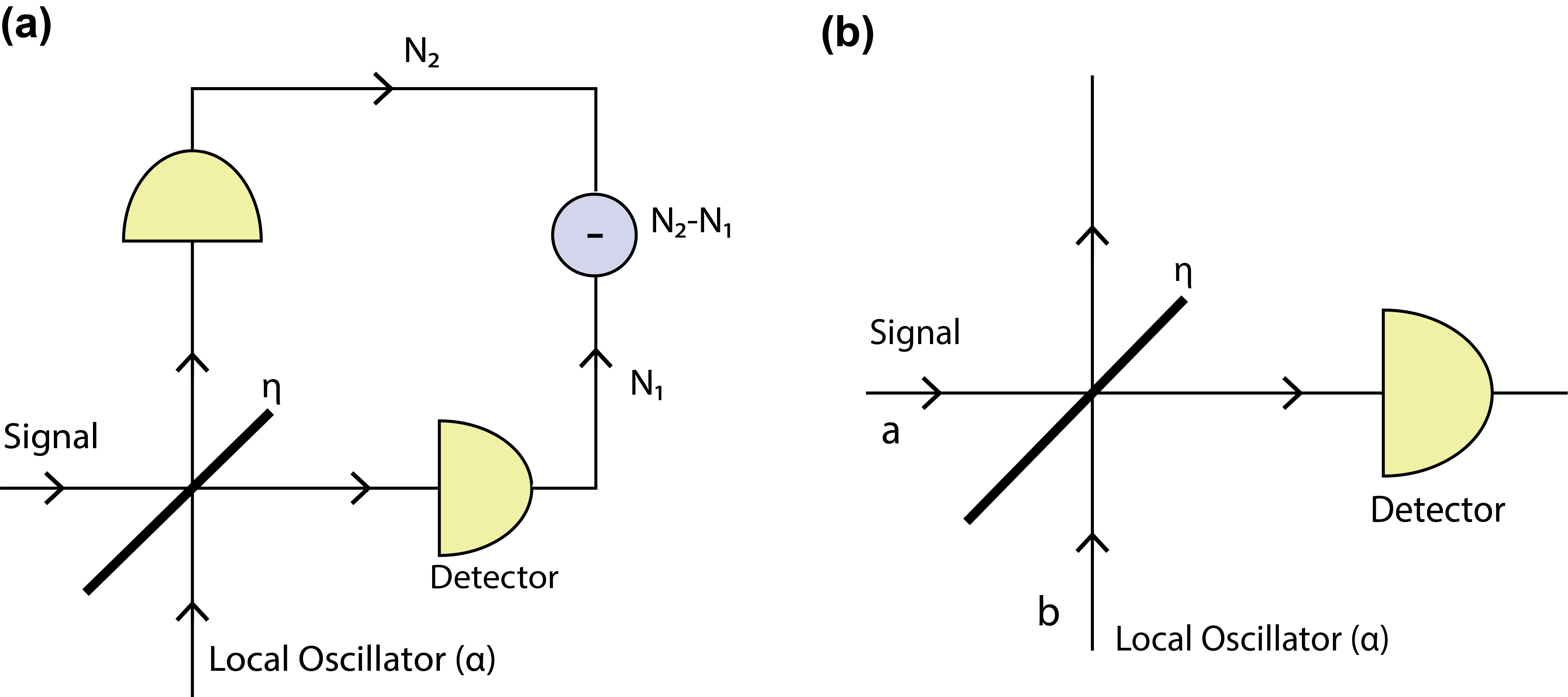}
    \caption{(a) Balanced homodyne detection: the unknown state signal interferes with a local oscillator (a coherent state $|\alpha\rangle$) and the intensities $N_1$ and $N_2$ at the output arms are measured. The difference in the intensities is proportional to the quadrature amplitude of the unknown state in the mode defined by the oscillator.(b) Photon-number resolving (PNR) measurement: the local oscillator is in a coherent state $|\alpha\rangle$, and detecting $N$ photon at the detector occurs with a probability $p(N|\alpha)=\langle N,\alpha|\rho|N,\alpha\rangle$. $\eta$ indicates the beam splitter in both cases.}
    \label{fig:homodyne}
\end{figure}

To cast homodyne tomography into shadow terminology, we first specialize the expansion in Eq.~(\ref{eq:rho-expansion}) to this case.
Substituting parameters $\mu \rightarrow (\theta,x_\theta)$, \newv{we get the expansion (see Eq.~(23) in Ref.~\cite{Raymer_Lvovsky}) }
\begin{equation}
    \rho=\int_{0}^{\pi}d\theta\,\int_{\mathbb{R}}dx_{\theta}\,{\textstyle \frac{1}{\pi}}\langle x_{\theta}|\rho|x_{\theta}\rangle F(x_{\theta},\theta)\,,
\end{equation}
where \(1/\pi\) corresponds to picking \(\theta\) uniformly from \([0,\pi)\), and \(\langle x_\theta|\rho|x_\theta\rangle\) to the conditional probability of the homodyne setup yielding \(x_\theta\) as the outcome. \newv{The matrix elements of the basis operators $\{F(x_\theta,\theta)\}$ have standard expansions in quantum optics literature. They are defined in terms of what are known as \textit{pattern functions} $f_{mn}(x)$~\cite{D'Ariano94,Raymer_Lvovsky}.
\begin{equation}
    \langle m|F(x_\theta,\theta)|n\rangle= e^{i(m-n)\theta} f_{mn}(x_\theta).
\end{equation}}
These pattern functions are symmetric in $m,n$, and are defined using Fock state wavefunctions $\psi_m(x)$ ($m^{th}$ energy eigen state of a harmonic oscillator) and $\phi_n(x)$ ($n^{th}$ non-normalizable solution of the Schr\"{o}dinger equation of a harmonic oscillator) as 
\begin{equation}
    f_{mn}(x) =\frac{\partial }{\partial x}\left(\psi_m(x)\phi_n(x)\right) \quad \text{for} \quad n\geq m.
\end{equation}

Because the quadrature probability distributions are normalized for each \(\theta\), the above equation can automatically be interpreted in the statistical fashion of Eq. (\ref{eq:rho-stat}). 
In this interpretation, basis operators \(F\) are sampled according to the probability distribution defined by \(\theta\) and \(\langle x_\theta|\rho|x_\theta\rangle\).
A regularized \textit{homodyne CV shadow} of size $T$, constructed using the samples $\{x_{\theta_i},\theta_i\}_{i=1}^T$, is therefore
\begin{equation}\label{eq:homshadows}
    \N{\sigma_T}=\, \frac{1}{T}\sum_{i=1}^T P_N F(x_{\theta_i},\theta_i) \,P_N~.
\end{equation}
\new{We remark that the randomized choice of $\theta_i$ independent of the state is what allows us to reinterpret this reconstruction method as an instance of classical shadows. We also remark that a homodyne shadow is distinct from the Wigner function projection of the state.  A shadow estimates the density operator in the Fock basis from a collection of homodyne samples with randomly chosen $\theta$. A Wigner function projection $\int dp_\theta W(x_\theta ,p_\theta ) =\langle{x_\theta}| \rho |{x_\theta}\rangle$, on the other hand, is the probability of obtaining specific measurement outcomes at a fixed values of $\theta$. }

We now adapt Lemma \ref{Lemma_1} to homodyne shadows, relying  on the bounds on pattern function matrix from Ref.~\cite{Guta_paper}. 
The following theorem gives an upper bound on the sample complexity of constructing shadows through homodyne samples, for an error $\|\N{\sigma_T}-\N{\rho}\|_\infty\leq \epsilon$ with high probability (greater than $1-\delta$).
\begin{theorem}\label{theorem:homodyne}
Using homodyne tomography, for $\epsilon,\delta \in (0,1)$, some positive constant $C_1$~\cite{Guta_paper}, and an N-dimensional homodyne shadow $\N{\sigma_T}$ (\ref{eq:homshadows}) of a single-mode CV state $\rho$,
\[{\rm Pr}(\|\N{\sigma_T}-\N{\rho}\|_\infty \leq \epsilon)\geq 1-\delta\] 
when the size of the shadow is at least
\[ T = \frac{2N^{13/3}C_1}{\epsilon^2}\left(\log\,1/\delta +\log\,2N\right).\]
\end{theorem}
\begin{proof}
From Theorem 3.2 of Ref.~\cite{Guta_paper}, we use  $\|\N{\sigma_T}\|^2_\infty \leq C_1 N^{7/3}$ and $\|\N{\sigma_T}\|^2_2\leq C_2N^3$, which \new{is valid for any value of homodyne angle $\theta$} and gives us 
\begin{equation}
    \|\N{\sigma_T}-\N{\rho}\|_\infty \leq \|\N{\sigma_T}\|_\infty+1\leq \sqrt{C_1}N^{7/6} + 1 = R
\end{equation}
and 
\begin{subequations}
\begin{align}
    \nu^{2}=\left\Vert \E_{{\cal T}}\,\sigma_{T}^{2}\right\Vert _{\infty}&=\left\Vert \int dx\,d\theta\,{\textstyle \frac{1}{\pi}}\langle x_{\theta}|\rho|x_{\theta}\rangle\sigma_{T}^{2}\right\Vert _{\infty}\\&\leq\int dx\,d\theta\,{\textstyle \frac{1}{\pi}}\langle x_{\theta}|\rho|x_{\theta}\rangle\|\sigma_{T}^{2}\|_{\infty}\\&\leq\int dx\,d\theta\,{\textstyle \frac{1}{\pi}}\langle x_{\theta}|\rho|x_{\theta}\rangle C_{1}N^{7/3}\\&\leq C_{1}N^{7/3}~.
\end{align}
\end{subequations}
Therefore, using Lemma \ref{Lemma_1}, the size of the shadow needs to be at least
\begin{align}
    T &= \frac{2N^2\left(\sigma^2+R\epsilon/2N\right)}{\epsilon^2}\left(\log\,2N-\log\,\delta\right) \\
    &\leq \frac{2N^2\left(C_1 N^{7/3}+1+\sqrt{C_1}N^{1/6}\epsilon/2\right)}{\epsilon^2} (\log\,2N+\log\,1/\delta).\nonumber
\end{align}
\end{proof}
We see that the upper bound on the minimum number of samples scales as $N^{13/3} \log\,N$ for a photon-cutoff $N$. 

Regarding the utility of homodyne shadows in calculating expectation values of observables, the pattern-function operators \(F\) are unequivocally more complex than the original qubit shadows~\cite{huang_kueng_preskill_2020}. Nevertheless, efficient methods are known to compute them using recursive relations~\cite{Richter2000}.

\subsection{Photon Number Resolving (PNR) Shadows}\label{subsec:PNR}

Our second application is based on the $T$-operator formalism developed in Ref.~\cite{Wunsche_1991}.
The relevant basis expansion of a single-mode CV state $\rho$ can be written in terms of what are called \textit{T-operators},
\begin{equation}\label{hetero_eq}
    \rho={\textstyle\frac{1}{\pi}}\int_{\mathbb{C}}d^2\alpha \,{\rm Tr}[\rho \,T(\alpha,\alpha^*)] \,\Bar{T}(\alpha,\alpha^*)~.
\end{equation}
Here, ${\rm Tr}[\rho T(\alpha,\alpha^*)]$ plays the role of a quasiprobability distribution over the complex plane $\mathbb{C}$, and the dual operators $\Bar{T}(\alpha,\alpha^*)$ are used to reconstruct the state given this distribution.

The $T$-operators are defined in terms of displaced Fock states, which are the result of applying a displacement operator $D(\alpha)=\exp(\alpha a^\dagger -\alpha^* a)$ on a Fock state:
$
    |\alpha,n\rangle \equiv D(\alpha)|n\rangle
$.
There exists a family of $T(\alpha,\alpha^*)$, and their corresponding duals $\Bar{T}(\alpha,\alpha^*)$, parameterized by $-1< r< 1$, such that the dual operator is given by
\begin{equation}
    \Bar{T}_r(\alpha,\alpha^*)=T_{-r}(\alpha,\alpha^*).
\end{equation}
$T(\alpha,\alpha^*)$ belonging to this family can be diagonalized as
\begin{equation}
    \begin{aligned}
        T_r(\alpha,\alpha^*)&= \sum_{n=0}^\infty \lambda^{(n)}_r D(\alpha)|n\rangle\langle n|D(\alpha)^\dagger \\
        &=\sum_{n}\lambda^{(n)}_r |\alpha,n\rangle\langle\alpha,n|~,
    \end{aligned}
\end{equation}
where the eigenvalues $\lambda^{(n)}_r$ are 
\begin{equation}\label{eq:Teigs}
    \lambda^{(n)}_r=\frac{2(-1)^n(1-r)^n}{\pi(1+r)(1+r)^n}~.
\end{equation}

Utilizing the above expansions, Eq.~(\ref{hetero_eq}) becomes
\begin{equation}\label{eq:integral}
    \rho=\int_{\mathbb C} \frac{d^2\alpha}{\pi}\sum_{n=0}^\infty \langle n,\alpha|\rho|n,\alpha\rangle \,\lambda^{(n)}_r \Bar{T}_r(\alpha,\alpha^*)~, 
\end{equation}
which we interpret as an expectation value of samples \(\lambda^{(n)}_r \Bar{T}_r(\alpha,\alpha^*)\) distributed according to the probabilities \(\frac{1}{\pi}\langle n,\alpha|\rho|n,\alpha\rangle\).
However, in a physical interferometric measurement procedure corresponding to this expansion [see Fig.~\ref{fig:homodyne}(b)], the \new{displacement} parameter \(\alpha\) comes from a non-compact set. This means that we need to first truncate the sample region in order to make the above a valid (i.e., normalizable) probability distribution.

We restrict the integration of Eq.~(\ref{eq:integral}) to be over a phase-space region \(\cal A\) of finite area \(A\). We can set this region to be, e.g., a disk of radius \(\alpha_{\text{max}}\). We will then be able to account for the truncated Fock space by setting \(\alpha_{\text{max}}^2 = N\) since that is the average occupation number of a coherent state with \(|\alpha|=\alpha_{\text{max}}\). In that case, \(A=4\pi N\) scales linearly with maximum photon number cutoff.

Projecting into an $N$-dimensional subspace implies that we consider only states with some max photon number \(N-1\) and regularized \(T\)-operators $\N{\Bar{T}_r}(\alpha,\alpha^*)$. Combining this with the truncation of the phase-space integration yields the following approximate expression for the truncated density matrix,
\begin{equation}\label{eq:approx-PNR}
    \N{\rho}\approx \int_{\cal A} \frac{d^2\alpha}{\pi A}\sum_{n=0}^{N-1} \langle n,\alpha|\rho|n,\alpha\rangle \,\, A \lambda^{(n)}_r \N{\Bar{T}_r}(\alpha,\alpha^*)~,
\end{equation}
in which we interpret $A \,\lambda^{(n)}_r\,\N{\Bar{T}_r} (\alpha,\alpha^*)$ as a \textit{PNR shadow} sampled from the probability distribution
\begin{equation}
    p(n,\alpha)=\frac{1}{A}\langle n,\alpha|\rho|n,\alpha\rangle
\end{equation}
with parameters \(0\leq n\leq N-1\) and \(\alpha \in\cal A\).

The eigenvalues $\lambda^{(n)}_r$ exhibit different behaviours for different values of $r$, and in Appx.~\ref{Appx_PNR} we analyze how many samples are required for density matrix estimation in each case. 
We obtain quartic scaling with \(N\) at \(r=0\), up to logarithmic corrections, and unfavorable exponential scaling for all other \(0<|r|<1\). 
In the favorable case, the upper bound on the minimum number of samples is
\begin{equation}
    T =  \frac{32 N^2\,A^2}{\pi^4 \epsilon^2}(\log\,2N+\log\,1/\delta) = \mathcal{O}(N^4 \log N)~,
\end{equation}
beating our theoretical estimates for homodyne shadows.
Therefore, we see that $T$-operators at $r=0$ have a favorable sample complexity scaling and can be used as a tomographic method to estimate an unknown state. 
The matrix elements of the $T$-operator in the Fock basis are equal to the matrix elements of the displacement operator up to a phase
\begin{equation}
\langle{m}|T_0(\alpha,\alpha^*)|{n} \rangle = (1)^{n}\langle{m}|D(2\alpha)|{n} \rangle,
\end{equation}
which can be used to efficiently compute the classical shadows.
Moreover, the case of $r=0$ can also be realized via photon-parity measurement tomography~\cite{Lutterbach97,Bertet02,Vlastakis13,Wang16}, providing another experimental technique to realize these shadows.

It is worth noting that in the limit $r=-1$, \(\Bar{T}_r(\alpha,\alpha^*)=T_{-r}(\alpha,\alpha^*)\) becomes a projector onto the coherent state $|\alpha\rangle$, making the measurement protocol equivalent to what is known as heterodyne measurement, achieved using a similar setup as homodyne measurement~\cite{Cives_Esclop_2000}. The basis expansion is given in terms of the Q-function $Q(\alpha) =\frac{1}{\pi} \langle{\alpha}|\rho|{\alpha}\rangle$.
However, the sample complexity upper bound diverges in this case.  It is possible to avoid this divergence by taking a limit $r \to -1$ as the number of samples increases.  Unfortunately, this method only achieves an upper bound that is exponential in $N$ because,
as noted above, for any $0<|r|  <1$, our sample complexity upper bound diverges exponentially in $N$ (see Appx. \ref{Appx_PNR}).

\section{Multimode Shadows}\label{section:multimode}

In this section, we show that it is possible to construct efficient representations of reduced density matrices on constant numbers of modes from local CV shadow data.

Multimode CV states are states that are made up of multiple \textit{modes} of the underlying system, such as different frequencies of light in an optical system. This is analogous to having a multi-qubit state in DV. Each mode is equipped with a generally continuous and non-compact space \(\mathcal{X}_i\) and \textit{M-mode}, possibly entangled, quantum states reside in \(\mathcal{X}=\bigotimes_{i=1}^M \mathcal{X}_i\). Let the space $\mathcal{X}$ be parameterized by a variable \(\bs{\mu}=(\mu^1,\ldots,\mu^M)\), and spanned by the set of basis operators \(\{\bs{\sigma}(\bs{\mu})\}\). Extending Eq.~(\ref{eq:rho-expansion}), the density matrix of a general multimode state belonging to this space can then be written as 
\begin{equation}\label{eq:multimode}
    \bs{\rho}=\int_\mathcal{X} d\bs{\mu}~ p(\bs{\mu})~ \bs{\sigma}(\bs{\mu}),
\end{equation}
where \(p(\bs{\mu})\) are the corresponding coefficients. Boldface symbols will be used to represent multimode variables and states henceforth. All the density operators and shadows considered in this section are the regularized versions, projected into the $N$-photon subspace of each mode. The superscripts indicating this are omitted for the sake of brevity.

Analogous to the  single-mode case in Sec.~\ref{section:single_mode_result}, we require a basis expansion of $\bs{\rho}$ such that \(p(\bs{\mu})\geq0~\forall\bs{\rho}\); the non-negativity of \(p(\bs{\mu})\) allowing it be interpreted as a probability distribution over $\bs{\mu}$. \new{We proceed by using an expansion of the form $p(\bs{\mu})=\langle \bs{\mu}|\bs{\rho}|\bs{\mu}\rangle$, where the quantum states or distributions $|\bs{\mu}\rangle$ parameterize the underlying space. For example, we can take $|\bs{\mu}\rangle$ to be local position state eigenstates in the case of homodyne tomography.}

\new{We construct shadows of multimode states using local, joint single-mode measurement.} For an estimator of a state $\bs{\rho}$ to qualify as a shadow, as introduced in Section Sec.~\ref{section:single_mode_result}, it must be equal to the state $\bs{\rho}$ in expectation. Below we show that it can indeed be satisfied for shadows constructed through local measurements. For this calculation, we utilize another expansion of a general quantum state $\bs{\rho}$ in terms of some set of separable operators \(\{\bs{\rho}_\alpha\}\) and corresponding complex coefficients \(\{c_\alpha\}\),
\begin{equation}\label{eq:op_expansion}
    \bs{\rho}=\sum_\alpha c_\alpha \bs{\rho}_\alpha=\sum_\alpha c_\alpha \bigotimes_{i=1}^M \rho_\alpha^i~.
\end{equation}
Such an expansion is possible irrespective of whether the state is separable or entangled.
The estimator becomes
\begin{subequations}
\begin{align}
    \E_{\bs{\mu}}~\bs{\sigma}(\bs{\mu})&=\int_{\mathcal{X}}d\bs{\mu}\times\langle\bs{\mu}|\bs{\rho}|\bs{\mu}\rangle\times\bigotimes_{i=1}^{M}\sigma(\mu^{i})\\&=\sum_{\alpha}c_{\alpha}\int_{\mathcal{X}}d\bs{\mu}\times\bigotimes_{i=1}^{M}\langle\mu^{i}|\rho_{\alpha}^{i}|\mu^{i}\rangle~\sigma(\mu^{i})\\&=\sum_{\alpha}c_{\alpha}\bigotimes_{i=1}^{M}\int_{\mathcal{X}_{i}}d\mu^{i}\langle\mu^{i}|\rho_{\alpha}^{i}|\mu^{i}\rangle\sigma(\mu^{i})\\&=\sum_{\alpha}c_{\alpha}\bs{\rho}_{\alpha}=\bs{\rho}~.
\end{align}
\end{subequations}

In the following theorem, we calculate the variance and thereby the sample complexity of estimating a multimode state through local CV shadows. This result reduces to the single-mode result from \ref{Lemma_1} when $M=1$.
\begin{theorem}\label{Theorem:multimode}
Fix $\epsilon, \delta \in (0,1)$, and let $\bs{\sigma}_T$ be an $N^M$-dimensional classical shadow obtained through local measurements $\bs{\sigma}(\bs{\mu})=\bigotimes_{i=1}^M \sigma(\mu^i)$ of an $M$-mode CV state $\bs{\rho}$ (\ref{eq:op_expansion}). If the size of the classical shadow is at least 
\begin{equation}\label{eq:multimodet}
    T = \frac{2N^{2M}(\nu^{2M}_1\sum_\alpha |c_\alpha|+R\epsilon/3N^M)}{\epsilon^2}(M\log2N+\log1/\delta),
\end{equation}
then
\[ {\rm Pr}\left( \|\bs{\sigma}_T-\bs{\rho}\|_\infty \leq \epsilon \right) \geq 1-\delta, \]
written in terms of the single-mode variance $\nu^2_1$ (\ref{eq:shadow-norms-nu}) and $R\geq \|\bs{\sigma}(\bs{\mu})-\E_{\bs{\mu}}\bs{\sigma}(\bs{\mu})\|_\infty$.
\end{theorem}
\begin{proof}
The measurement variance for an $M$-mode state proceeds similar to the single-mode case with the help of Eq. (\ref{eq:op_expansion}), 
\begin{subequations}
 \begin{align}
       \nu_{M}^{2}&=\left\Vert \E_{\bs{\mu}}~\left(\bs{\sigma}(\bs{\mu})\right)^{2}\right\Vert _{\infty}\\&=\left\Vert \int_{{\cal X}}\prod_{j=1}^{M}d\mu^{j}\times\langle\bs{\mu}|\bs{\rho}|\bs{\mu}\rangle\times\bigotimes_{i=1}^{M}\left(\sigma(\mu^{i})\right)^{2}\right\Vert _{\infty}\\&=\left\Vert \sum_{\alpha}c_{\alpha}\bigotimes_{i}\int_{{\cal {\cal X}}_{i}}d\mu^{i}\langle\mu^{i}|\rho_{\alpha}^{i}|\mu^{i}\rangle\left(\sigma(\mu^{i})\right)^{2}\right\Vert _{\infty}\\&\leq
       \sum_{\alpha}|c_{\alpha}|\,\prod_{i}\left\Vert \int d\mu^{i}\langle\mu^{i}|\rho_{\alpha}^{i}|\mu^{i}\rangle\left(\sigma(\mu^{i})\right)^{2}\right\Vert _{\infty}\\&\leq\sum_{\alpha}|c_{\alpha}|\,\prod_{i}\nu_{1}^{2} \\&\leq\nu_{1}^{2M}\sum_{\alpha}|c_{\alpha}|~.
 \end{align}
\end{subequations}
Using matrix Bernstein inequality and following a calculation similar to Lemma.~\ref{Lemma_1}, replacing $N$ with $N^M$, we get the desired result in Eq. (\ref{eq:multimodet}).
\end{proof}

\begin{figure*}[t]
     \centering
    \includegraphics[scale=0.57]{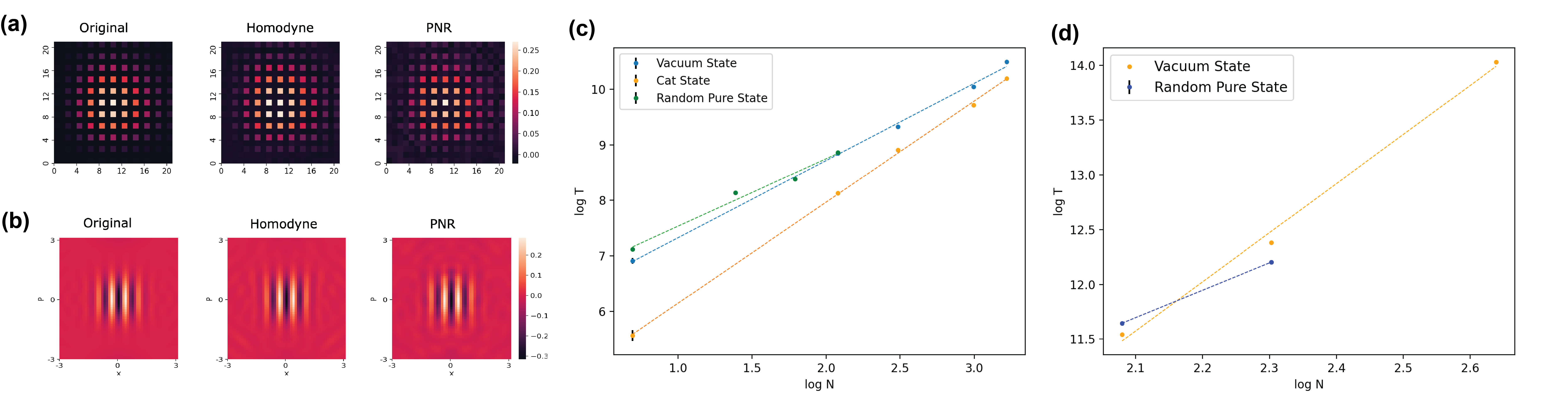}
        \caption{Single-mode shadows: (a) Density matrix elements and (b) Wigner function representation of the original state, a homodyne shadow and a PNR shadow of an even CSS state $|\psi\rangle \propto |\alpha\rangle +|-\alpha\rangle$ with $\alpha=\sqrt{10}$. (c-d) Numerical sample complexity of single-mode shadows; precision and confidence parameters are $\epsilon=0.1$ and $\delta=0.1$. (c)~Homodyne Shadows: Vacuum state ($T\propto N^{1.389}$), CSS state with amplitude 
        $\alpha=\sqrt{10}$ ($T\propto N^{1.820}$), and random pure states ($T\propto N^{1.212}$). (d)~PNR Shadows: Vacuum state ($T\propto N^{3.51}$) and pure random states.}
       \label{fig:combined}
\end{figure*}

As a consequence, we obtain a sample complexity bound on estimating $k$-local reduced density matrices, \new{which are the outcome of tracing over all but k modes of a multimode state. Since we are tracing out some of the modes, we are left with a ‘reduced density matrix’.}

\begin{corollary}\label{corollary:klocal}
If $\bs{\rho}_k$ denotes the reduced density matrix of $\bs{\rho}$ in $k\leq M$ modes, and its estimator constructed through local measurements is  $\bs{\sigma}_{T,k}$, such that the size of the shadow is at least 
\[T = \frac{2N^{2k}(\nu^{2k}_1\sum_\alpha |c_\alpha|+R\epsilon/3N^k)}{\epsilon^2}(k \log2N+\log1/\delta)\]
then
\[ {\rm Pr}\left( \|\bs{\sigma}_{T,k}-\bs{\rho}_k\|_\infty \leq \epsilon \right) \geq 1-\delta. \]
\end{corollary}
\begin{proof}
Without loss of generality, let's assume $\bs{\rho}_k$ to be the reduced density matrix over the modes $1,\ldots,k$, with $k\leq M$, i.e. $\bs{\rho}_k={\rm Tr}_{k+1,..,M}(\bs{\rho})$. Define $\bs{\mu}_k=(\mu^1,\ldots,\mu^k)$ to parameterize this $k$-mode subspace, and the local operators
\[\bs{\sigma}_{k}(\bs{\mu}_k)={\rm Tr}_{k+1,..,M}(\bs{\sigma}(\bs{\mu}))=\bigotimes_{i=1}^k \sigma(\mu^i).\]
The classical shadow of the reduced matrix $\bs{\sigma}_{T,k}$ is then the average of $T$ such operators. The variance of estimating this reduced matrix, following earlier calculations, is
\begin{align}
    \nu_{k}^{2}=\left\Vert \E_{\bs{\mu}_{k}}\left(\bs{\sigma}_{k}(\bs{\mu}_{k})\right)^{2}\right\Vert _{\infty}\leq\nu_{1}^{2k}\sum_{\alpha}|c_{\alpha}|~.
\end{align}
\end{proof}

One can further generalize this result to derive sample complexity bounds for determining the reduced density matrix on any choice of $k$-local modes.  The proof follows from an application of the union bound and increases the sample complexity by a factor of $k$~\cite{huang_kueng_preskill_2020,huang2021provablyefficient}.

\newv{The sum of the absolute values of the coefficients in Eq.~\ref{eq:op_expansion} is a factor that affects the sampling complexity of multimode states. These coefficients can in general be complex. 
When the state is separable (i.e. can be written as a product state) this factor becomes 1, thereby resulting in a complexity that is equivalent to sampling all the modes individually, as expected of a product state. When the coefficients are all positive, they again sum up to 1 (since the trace of the state must equal 1), resulting in no increase in complexity as compared to a separable state.  In the general case, we can still upper bound the sampling complexity, but the overall scaling with the number of modes has a larger prefactor in the exponent.  These results are summarized in Table.~\ref{tab:protocols}. }

\newv{We would like to add that the proposed local-joint measurements are not the most general measurements for multimode CV states. Upon trying many different bases of entangled/global states, we have not found a suitable one that fulfills our criteria of a positive weight function for all states ($p(\mu)$ in Eq.~\ref{eq:multimode}), but have not ruled out the possibility of their existence.  Studies of whether global measurements can facilitate such a statistical interpretation, and of the sample complexity when using such bases are interesting avenues for future work.}

\begin{corollary}
    \newv{Similarly, we can estimate the expectation values for a list of multimode observables $\{O_1,\ldots,O_M$\}, local to $\{k_1,\ldots,k_M\}$ number of modes respectively and each mode within an $N$-photon subspace, from a shadow of size upper bounded by $T=\mathcal{O}\left(\max_i \frac{\|O_i\|}{\epsilon^2} \log (M/\delta)\right)$
    such that
    \begin{align}
        \Pr\left(|Tr(\bs{O}_i \bs{\rho}) - Tr(\bs{O}_i \bs{\sigma}_T)|\leq \epsilon\right) \geq 1-\delta \quad \forall\, 1\leq i\leq M.
    \end{align}}
\end{corollary}
\begin{proof}
    \newv{Since the 1-norm of matrix is the sum of its singular values, and using the Holder's inequality,
    \begin{equation}
        |Tr(\bs{O} \bs{\rho})| \leq \|\bs{O} \bs{\rho}\|_1  \leq \|\bs{O}\|_1 \|\bs{\rho}\|_\infty.
    \end{equation}
    Therefore,
    \begin{equation*}
        \begin{aligned}
        Pr \left(|Tr(\bs{O}_i(\bs{\rho}-\bs{\sigma}_T))| \leq \epsilon\right) &\leq Pr\left(\|\bs{O}\|_1 \|\bs{\rho}\|_\infty \geq \epsilon\right) \\
        &=Pr\left( \|\bs{\rho}\|_\infty \geq \frac{\epsilon}{\|\bs{O}\|_1}\right).
    \end{aligned}
    \end{equation*}
    Using matrix Bernstein inequality, this implies
    \begin{equation*}
    \begin{aligned}
        Pr (|Tr(&\bs{O}_i(\bs{\rho}- \bs{\sigma}_T))| \geq  \epsilon ) \\&\leq 2N^{k_i} \,\exp\left(-\frac{T\left(\epsilon/\|\bs{O}\|_1\right)^2}{2N^{2k_i}(\nu^{2k_i}+R\epsilon/3N^{k_i}\|\bs{O}\|_1)}\right)~,
    \end{aligned}
    \end{equation*}
    where is $\nu$ and $R$ are as defined in Theorem.~\ref{Theorem:multimode}.
    Equating the upper bound on probability to $\delta/M$, and using the union bound \cite{huang_kueng_preskill_2020}, a shadow size upper bounded by 
    \begin{equation}
        T = \max_i \, \frac{2\|\bs{O}_i\|_1 N^{2k_i}\nu^{2k_i}}{\epsilon^2}\left(k_i\log\,2N+\log\,M/\delta\right)
    \end{equation}
    suffices to estimate the list of multimode observables $\{\bs{O}_i\}_{i=1}^M$ with the desired accuracy.
    }
\end{proof}

\section{Numerical and Experimental tests}\label{section:numerical_results}

In this section, we verify our shadow construction methods and present numerical results on sample complexity scaling, with the maximum occupation number \(N\), of single-mode shadows. We also analyze experimental homodyne data and demonstrate the reconstruction of multimode states.

\subsection{Single mode simulations} \label{subsec:single_mode_simulations}

To analyze the applicability of our theoretical guarantees, we construct estimators of density matrices of certain target states using homodyne and PNR shadows.  \new{The  measurement bases in all cases were generated using uniform pseudorandom number generators.}
We then analyze how close they are to the target states. 
Figure~\ref{fig:combined}(a) shows how a reconstructed CSS state with mean-photon number $10$ compares to the original density matrix elements for both homodyne and PNR tomography. A comparison in terms of the Wigner function is shown in Fig.~\ref{fig:combined}(b). 
To achieve equally accurate shadows such that $\|\N{\sigma_T}-\N{\rho}\|_\infty=0.1$, we had to use $5\times10^4$ and $10^6$ samples for homodyne and PNR, respectively.

Of primary interest in the classical shadow formalism is the number of samples required to achieve a given target precision with high-probability.  
 
\newv{This minimum shadow size for a given state $T$ was numerically obtained by repeating the estimation process for several shadow sizes and identifying the one which results in a precision $\epsilon$, with a probability of at least $1-\delta$. We observed that if we are looking for a fixed error in norm (i.e. fixing $\epsilon$), increasing the shadow size decreases the probability with which this error occurs (i.e. $\delta$) as the estimation becomes more accurate. In fact, we observed that $T$ for a fixed $\epsilon$ is proportional to $\log{1/\delta}$ and this helped us narrow down our search for minimum shadow sizes corresponding to our chosen pair ($\epsilon$, $\delta$)=(0.1, 0.1). The scaling with the occupational number cutoff is then obtained by finding the minimum shadow size for different $N$, and the error bars indicate the statistical errors in estimating $T$ via this process. In Fig.~\ref{fig:combined}~(c)-(d), we plotted this observed minimum size for different states, for homodyne and PNR shadows respectively.}
\\

We ran our simulations for a vacuum state $|\psi\rangle=|0\rangle$, a CSS state $|\psi\rangle\propto |\alpha\rangle+|-\alpha\rangle$ with mean photon number $|\alpha|^2=10$, and random pure states. The random pure states of size $N$ had non-zero support only up to $(N-1)^{\text{th}}$ Fock state. The results show that homodyne has a scaling better than $N^2$ for the cases we tried, consistently outperforming the theoretical upper bound. \newv{For  vacuum states, PNR on the other hand has a scaling much closer to the predicted upper bound of $N^4$.  We also tested the scaling of PNR tomography of random pure states, but were computationally limited in the accessible sizes.  Over the range we observed, the results are consistent with the scaling of the vacuum state.}

\subsection{Experimental data}\label{section:experiment}

Here, we analyze the sample complexity of experimental homodyne data from Ref.~\cite{NIST_experiment}. The experiment consisted of creating \new{Coherent State Superpositions} (CSSs), where the created states were approximations to small-amplitude, even CSS states $|\psi\rangle\propto|\alpha\rangle+|-\alpha\rangle$.  
Squeezed vacuum was generated through spontaneous parametric down-conversion in a \(\text{KNbO}_{3}\) non-linear crystal: an up-converted laser pulse at $430$~nm created the squeezed vacuum at $860$~nm. After spectral filtering, the squeezed vacuum was incident on a weakly reflecting beam splitter with reflectivity $R$. Photons detected in the reflected path herald a CSS emerging from the transmitted port of the beam splitter. The resulting CSS was then directed to a homodyne detection setup. The homodyne setup consisted of a 50/50 beam splitter, two high-efficiency photodiodes and a low-noise amplification circuit.  \new{The local oscillator was created with a strong laser field ($\approx 10^9$~photons per pulse) that allows access to the quadratures through balanced detection.} The heralded photon detection was done using either two single-photon avalanche diodes (SPADs)~\cite{SPAD} or one transition edge sensor (TES)~\cite{TES}. The experiment analyzed here used both heralded-photon detection schemes. The CSSs were heralded upon the detection of two heralding photons, generating the approximation of a small-amplitude even cat state.

\begin{figure*}[t]
     \centering
     \includegraphics[scale=0.6]{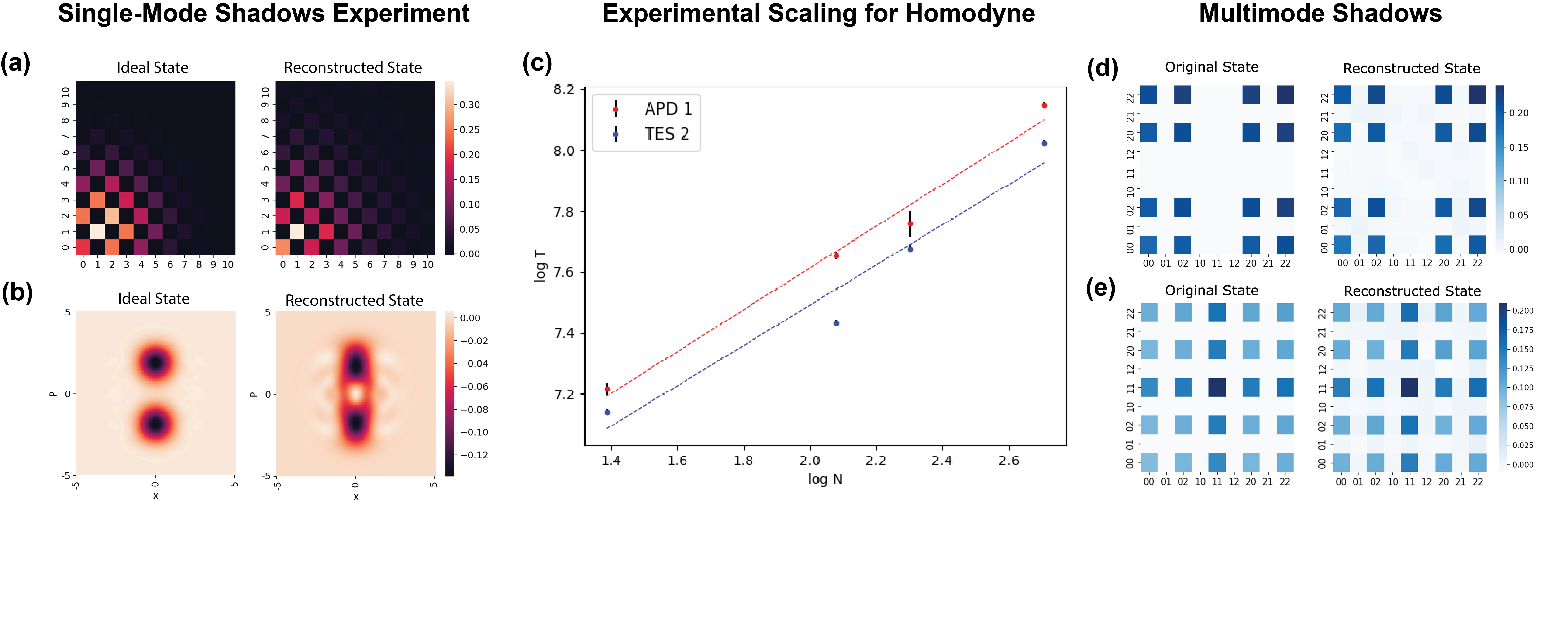}
     \caption{Experimental test of single-mode shadows:  (a) Density matrix elements and (b) Wigner function representation of the \new{ideal} CSS state \new{from theory} (left) and the reconstructed state \new{from experimental} (right) homodyne measurements. \new{Ideal} CSS state: $|\psi\rangle \propto |\alpha\rangle+|-\alpha\rangle$, $\alpha=1.32$, $T=3.10^5$. (c) Sample complexity scaling with maximum photon number $N$, from the homodyne data of the experiment. The two methods `TES2' and `APD1' correspond to photon-detection through transition edge sensor and single-photon avalanche diode, respectively. The scaling shown is for $\epsilon=0.4$ and $\delta=0.1$. (d-e) Homodyne two-mode shadows: Density matrix elements of the original states and numerically reconstructed states. The number of samples $T=5.10^4$, $N=2$ and $\alpha=\sqrt{1.5}$ in both cases. (d)~2-mode separable state: $|\Psi\rangle = |\psi_\alpha\rangle\otimes |\psi_\alpha\rangle$, where $|\psi_\alpha\rangle\propto |\alpha\rangle+|-\alpha\rangle$, and the infinity norm distance of the reconstructed state from the ideal state is $\epsilon=0.03$. (e)~2-mode entangled state: $|\Psi\rangle \propto |\alpha,\alpha\rangle+|-\alpha,-\alpha\rangle$, and $\epsilon=0.05$. 
     }
     \label{fig:multimode_verification}
\end{figure*}

Using the homodyne data from this experiment, we constructed an estimate of the state, an example of which is shown in Fig.~\ref{fig:multimode_verification}(a-b), and studied how close the reconstructed states are to the target states. \new{Before analyzing the sample complexity, the experimental data was pre-processed to account for propagation losses and detection inefficiencies  \cite{NIST_experiment}.} 

\newv{We then repeated the same numerics as in Sec.~\ref{subsec:single_mode_simulations} to find the minimum shadow size ($T$) required to achieve a fixed ($\epsilon$, $\delta$) for different occupational number ($N$) cutoffs. Working with a given dataset means our shadow size is limited to the size of the dataset and that accuracy can’t be improved arbitrarily. Fig.~\ref{fig:multimode_verification}(c) shows how $T$ changed with $N$, and we observe a scaling of $T \propto N^{0.659\pm 0.001}$ for a \new{\st{cat} CSS} with $\alpha=1.16$ using the TES, and as $T\propto N^{0.685\pm 0.001}$ for a \new{\st{cat} CSS} with $\alpha=1.32$ using the SPADs. This scaling of approximately $T\propto N^2$ aligns with our numerical simulations of single-mode homodyne shadows,  and is a further testament to the reliability of our simulations.
}

\subsection{Multimode simulations}
 
We now provide an example reconstruction of both separable and entangled two-mode states.

In Fig.~\ref{fig:multimode_verification}(d) we can see original and reconstructed density matrices of a separable 2-mode state $|\Psi\rangle = |\psi_\alpha\rangle\otimes |\psi_\alpha\rangle$, where $|\psi_\alpha\rangle$ is an even CSS state of amplitude $\alpha=\sqrt{1.5}$. In Fig.~\ref{fig:multimode_verification}(e), we have the reconstruction of an entangled state $|\Psi\rangle=|\alpha\rangle \otimes |\alpha\rangle+|-\alpha\rangle\otimes|-\alpha\rangle$, with the same number of measurements. The infinity-norm distance $\epsilon$ between the original and reconstructed states is 0.03 and 0.05, respectively.  The increased sample complexity needed to reach a given precision for the entangled state tomography is consistent with our general expectations from the theoretical analysis in Sec.~\ref{section:multimode}.

\section{Outlook and Conclusions}\label{section:conclusions}
We developed a formalism to derive worst-case sample complexity of a large class of  state-reconstruction procedures for single-mode continuous-variable (CV) states \new{(Lemma \ref{Lemma_1}, Theorem \ref{theorem:homodyne})}. In an experimental context, our work provides a tool to compare the sample complexity of different measurement methods. For photon-number resolving measurements, we show through numerical simulations that our analytical derived upper bound is close to the observed complexity. In the case of homodyne tomography, we notice that the simulation results surpass our derived bound, rendering homodyne tomography more efficient than PNR and photon-parity tomography in practice.

We obtain the sample complexity of estimating multimode CV states by reconstructing shadows of the true states from local measurements on each mode, in the spirit of the qubit-based shadow tomography proposal~\cite{aaronson2018shadow,huang_kueng_preskill_2020,Elben2022}. 
We also consider a CV manifestation of a global version of the original shadow protocol, where one uses global measurements over varying linear combinations of the constituent modes for state reconstruction. Of particular note is that, in contrast to the discrete-variable systems, our analysis does not rely on the construction of state designs (approximate or otherwise) such as in~\cite[Sec.~VI.A]{Iosue22}.   Instead, we focus on adapting existing techniques for CV tomography to fit within the framework of randomized measurements.  

In summary, we recast existing homodyne and photon-number-resolving protocols as shadow-tomography protocols and show that such protocols yield good local estimates of a multimode state using a number of samples that is polynomial in the number of modes \new{\ref{corollary:klocal}}. 
Our CV shadow framework can be extended to analyze other metrics like the total variation distance between two outcome probability distributions, which is valuable for verifying sampling experiments like Boson sampling~\cite{Hangleiter22}. 
Another remaining open question is to determine robustness of CV protocols to noise, e.g.,  within the framework of robust shadow estimation~\cite{Chen21}.  Deriving rigorous lower bounds that can match the numerically and experimentally observed scaling for homodyne tomography is also an important outstanding challenge.

\begin{acknowledgments}
We thank Yi-Kai Liu and Dominik Hangleiter  for helpful discussions. MJG and VVA thank Joe Iosue and Kunal Sharma for discussions, as well as collaborations on~\cite{Iosue22}.  We thank  S. Becker, N. Datta, L. Lami and C. Rouz{\' e} for  helpful discussions about the relations of our results to \cite{Becker22} and for pointing out a change in our homodyne bounds proof that allows an improvement from  $\mathcal{O}(N^5)$ to $\mathcal{O}(N^{4+1/3})$ scaling. VVA thanks Olga Albert and Ryhor Kandratsenia for providing daycare support throughout this work. This work is supported in part by NIST grant 70NANB21H055{\_}0 and NSF QLCI grant OMA-2120757. 

\end{acknowledgments}
\vspace{5pt}
\appendix 
\section{Scaling with PNR Tomography}\label{Appx_PNR}
Beginning with the approximate expression for a truncated density matrix in the displaced Fock basis (\ref{eq:app_pnr}), we want to estimate the sample complexity of PNR shadows.
\begin{equation}\label{eq:app_pnr}
    \N{\rho}\approx \int_{\cal A} \frac{d^2\alpha}{\pi A}\sum_{n=0}^{N-1} \langle n,\alpha|\rho|n,\alpha\rangle \,\, A \lambda^{(n)}_r \N{\Bar{T}_r}(\alpha,\alpha^*)~
\end{equation}
The above expansion in terms of \(T_r(\alpha,\alpha^*)\) is valid for \(-1< r< 1\) as mentioned in Sec.~\ref{subsec:PNR}. Depending on whether $r$ is positive or negative, one of $\lambda^{(n)}_r$ or $\lambda^{(n)}_{-r}$ (the latter being an eigenvalue of $T_{-r}(\alpha,\alpha^*)$) increases with $n$ and the other decreases with $n$. For $r=0$, the eigenvalues $\lambda^{(n)}_0$ are independent of $n$. Using the triangle inequality and Eq.~(\ref{eq:Teigs}), the shadow norm bound (\ref{eq:shadow-norms-R}) becomes
\begin{subequations}\label{eq:app_R}
\begin{align}
    \|A \lambda^{(n)}_r \N{T_{-r}} - \N{\rho} \|_\infty &\leq A  \|\lambda^{(n)}_r \N{T_{-r}} \|_\infty +A \\
    &\leq A|\lambda^{(N)}_{-|r|} \, \lambda^{(0)}_{|r|}|+A 
    \\
    &\leq \frac{2A|\lambda^{(N)}_{-|r|}| }{\pi(1+|r|)}+A \equiv R,
\end{align}
\end{subequations}
with the equality holding at $r=0$. As for the variance from Eq.~(\ref{eq:shadow-norms-nu}), we have 
\begin{subequations}\label{eq:app_var1}
\begin{align}
    \nu^{2}&=\left\Vert \E_{n,\alpha}\left(A\lambda_{r}^{(n)}\N{T_{-r}}\right)^{2}\right\Vert _{\infty}\\&=\left\Vert \int_{{\cal A}}\frac{d^{2}\alpha}{\pi A}\,\sum_{n=0}^{N-1}p(n,\alpha)\left(A\lambda_{r}^{(n)}\N{T_{-r}}\right)^{2}\right\Vert _{\infty}\\
    &\leq\int_{{\cal A}}\frac{d^{2}\alpha}{\pi A}\,\sum_{n=0}^{N-1}p(n,\alpha)\,A^2 \left\Vert (\lambda_{r}^{(n)}\N{T_{-r}})^{2}\right\Vert _{\infty}~\\
    &\leq\int_{{\cal A}}\frac{d^{2}\alpha}{\pi A}\,\sum_{n=0}^{N-1}p(n,\alpha)\,A^2|\lambda^{(0)}_{|r|}\lambda^{(N)}_{-|r|}|^2\\
    &\leq A^2 |\lambda^{(0)}_{|r|}\lambda^{(N)}_{-|r|}|^2 
\end{align}
\end{subequations}
The remaining norm can be bounded using the formula (\ref{eq:Teigs}) as was done in Eq. (\ref{eq:app_R}), and the remaining sum and integral can be upper-bounded by a factor of \(A\), yielding
\begin{equation}\label{eq:app_var2}
    \nu^{2}\leq\frac{16\,A^2}{\pi^{4}(1-|r|^2)^{2}}\left(\frac{1+|r|}{1-|r|}\right)^{2N}~.
\end{equation}

Applying Lemma \ref{Lemma_1}, and using \(A = a N\), the upper bound on minimum number of samples is
\begin{equation}
    T = \frac{32\,a^2\, N^4}{\pi^4\epsilon^2(1-|r|^2)^2}\left(\frac{1+|r|}{1-|r|}\right)^{2N} \log\left(\frac{2N}{\delta}\right)~,
\end{equation}
growing exponentially in \(N\) for $r\neq 0$. 

\bibliographystyle{apsrev-nourl-title.bst}
\bibliography{main.bib}

\begin{thebibliography}{47}
\bibitem[{ter()}]{terminology}
\bibinfo{note}{Note, the words shadow, estimator and reconstructed state are
  used interchangeably throughout the paper.}
\bibitem[{smo()}]{smoothcutoff}
\bibinfo{note}{For example, for a regulator $R$ one can find a sufficiently
  large $N$ such that $||P_N R\rho R^\dag P_N - R \rho R^\dag|| < \epsilon$ for
  all density matrices $\rho$ and apply our results directly based on the
  triangle inequality. More generally, the bounds we prove on the projected
  shadows could be adapted to regulators with smooth cutoffs, but we leave this
  approach for future work.}
\end{thebibliography}

\end{document}